\documentclass[a4paper,fleqn]{cas-dc}
\usepackage[authoryear]{natbib}
\newtheorem{theorem}{Theorem}
\newtheorem{definition}{Definition}
\newtheorem{lemma}{Lemma}
\newdefinition{remark}{Remark}
\newproof{proof}{Proof}
\newproof{pot}{Proof of Theorem \ref{thm}}

\usepackage[commandnameprefix=ifneeded]{changes}
\usepackage{graphics} 
\usepackage{hyperref} 
\usepackage{mathrsfs}

\usepackage{amsmath,amssymb,amsfonts}
\usepackage{algorithmic}
\usepackage{graphicx}
\usepackage{algorithm,algorithmic}
\usepackage{siunitx}

\DeclareMathOperator{\interior}{int}

\usepackage{tikz}
\usepackage{pgfplots}
\pgfplotsset{compat=newest}
\usetikzlibrary{external}
\newlength\figH
\newlength\figW
\newcommand{\CS}{\mathcal{X}}       
\newcommand{\TS}{\mathcal{X}_T}       
\newcommand{\RA}[2]{R_{[#1]}\left(#2\right)} 
\newcommand{\uu}{\mathbf{u}}
\newcommand{\ww}{\mathbf w}

\newcommand{\Uset}{\mathcal{U}}
\newcommand{\Tbdl}[1]{\mathcal{T}_{#1}}

\newcommand{\Feas}{\mathcal{R}}

\newcommand{\clf}[1][\empty]{\ifx#1\empty V \else V(#1)\fi}
\newcommand{\cbf}[1][\empty]{\ifx#1\empty h \else h(#1)\fi}

\newcommand{\Vap}{\hat V}
\newcommand{\Rap}{\hat{\mathcal{R}}}

\newcommand{\uuinfhat}{\hat{\kappa}_\partial}

\usepackage{changes}
\definechangesauthor[color=purple]{TC}
\definechangesauthor[color=blue]{JO}
\definechangesauthor[color=olive]{JO2}

\def\tsc#1{\csdef{#1}{\textsc{\lowercase{#1}}\xspace}}
\tsc{WGM}
\tsc{QE}
\tsc{EP}
\tsc{PMS}
\tsc{BEC}
\tsc{DE}

\begin{document}
\let\WriteBookmarks\relax
\def\floatpagepagefraction{1}
\def\textpagefraction{.001}
\shorttitle{}
\shortauthors{J. Olucak and T. Cunis}

\title [mode = title]{Input-to-state Stable Approximate Nonlinear Model Predictive Control with Realtime
Feasibility}  
       
\tnotemark[1] 

\tnotetext[1]{This research is partially supported by the Ministry of Science, Research and Arts of the state of Baden-Württemberg under funding number MWK32-7531-49/13/7 for the project DaSO: Data-driven Spacecraft Operations.} 

%

\author[1]{Jan Olucak}[orcid=0009-0008-2210-8692]

\cormark[1]


\ead{jan.olucak@ifr.uni-stuttgart.de}

\ead[url]{}

\credit{Data curation, Investigation, Visualization, Conceptualization, Software, Writing – original draft}

\affiliation[1]{organization={Institute of Flight Mechanics and Controls, University of Stuttgart},
            addressline={Pfaffenwaldring 27}, 
            city={Stuttgart},
            postcode={70569}, 
            state={Baden-Württemberg},
            country={Germany}}

\author[1]{Torbjørn Cunis}[orcid=0000-0003-1672-4321]


\ead{tcunis@ifr.uni-stuttgart.de}

\ead[url]{}

\credit{Conceptualization,  Methodology, Supervision, Writing – review \& editing}


\cortext[1]{Corresponding author}



\begin{abstract}
In this paper, a {computational{ly} lightweight} {approximate} robust nonlinear model predictive control (NMPC) law is proposed based on {a pair of} input-to-state control Lyapunov function and robust control barrier function. The result builds upon and augments a recently introduced nominal infinitesimal-horizon NMPC scheme which permits small-sized quadratic programs to compute the feedback law for nonlinear constraint systems {on embedded hardware} in real time. {Numerical experiments} for nonlinear constrained spacecraft control {and comparison} to other robust NMPC schemes from the literature {demonstrate the effectiveness of the proposed scheme.}
\end{abstract}




\maketitle

\section{Introduction}
Nonlinear model predictive control (NMPC) is a well-established control method for constrained (non)linear systems because of its inherent property of seeking an optimal solution while assuring state and control constraints, (complex) path constraints and providing (asymptotic) stability~\citep{grune2017}. 
Whereas the closed-loop nature of the nominal NMPC has some inherent robustness properties~\citep{limon2009,yu2014,allan2017,Raković2019}, these are very limited for small disturbances in the constraint case. If disturbances are not explicitly considered in the prediction model, this might lead to constraint violations and stability might not be ensured anymore. Robust NMPC (RMPC) techniques~\citep{bemporad1999,Raković2019} with explicit consideration of disturbances and realtime feasibility are still an active field of research. 

To address the issue of RMPC, various strategies have been discussed in the literature:
Min-max MPC~\citep{raimondo2009} approaches try to account for the worst case disturbance. It provides rigorous theoretical guarantees, but the resulting infinite-dimensional min-max optimization problem can be hardly solved. Approximations consider, for example, only a finite number of uncertainty realization~\citep{scokaert1998} or using scenario-based approaches to give probabilistic guarantees~\citep{calafiore2013,fagiano2015}. However, these approaches are computationally intractable for large number of samples and long horizons. 

Tube-based (N)MPC~\citep{mayne2005} is based on the idea of computing a tube or funnel around the nominal trajectory that ensures robust constraint satisfaction. Here, the constraints are tightened by robust invariant sets. Early approaches used two NMPC schemes~\citep{mayne2011} which is computationally {intractable}.
Another direction computes over-approximations of the forward reachable set for constraint tightening simultaneously to the actual NMPC problem~\citep{bravo2006}, which is again computationally too heavy. For this approach, \citep{schurmann2018} proposed to split the problem into the NMPC optimization and reachable set estimation, to reduce the computational demand. 
 For improved computational effort, offline schemes are often used \citep[see, e.g.,][]{yu2013,kohler2018,rakovic2023}, which have similar computational effort as nominal NMPC, but might be conservative due to fixed control laws. 
To reduce conservatism, approximations of the true reachable set using ellipsoidal sets have been proposed by \citep{zanelli2021}. The computational complexity is the same as the nominal NMPC while ensuring robustness. {Note that for large scale systems or with long prediction horizons, the computational effort of nominal NMPC {is often} already prohibitive for real-time applications.}
The approach of \citep{zanelli2021} can be combined with the well-known real-time iteration (RTI) scheme \citep{frey2024} to further improve the computational effort for real-time execution.
To  reduce conservatism of tube MPC, some approaches jointly optimize for a nominal trajectory and a (linear) feedback law. For instance, \citep{messerer2021} solves for the nominal trajectory and controller in an alternating fashion. 
A similar approach is system level synthesis to compute the exact reachable set \citep{anderson2019,sieber2022a} for linear systems and was extended to nonlinear systems with robust constraint satisfaction by \citep{leeman2025b}. The joint optimization can reduce conservatism and improves performance, but the resulting optimization problems are large and hence are a major obstacle for real-time implementation. To alleviate this problem, custom solvers \citep[e.g.,][]{leeman2024} for linear system can be used. Recently, the approach of \citep{leeman2024} was combined with sequential convex programming and an RTI scheme for both the outer sequential loop and inner convex program to tackle nonlinear systems~\citep{leeman2025a}.
While substantial improvement can be observed for tube MPC, the computational effort might be prohibitive for applications with large scale systems, long prediction horizons and limited computing resources, unless an RTI scheme such as {those given by} \citep{frey2024} or \citep{leeman2025a} is applied. Yet, even an RTI scheme might not be sufficient for computationally constrained applications.

{Another direction to allow real-time computations, are} computationally low demanding schemes that rely on precomputed sequences of invariant sets to reduce the prediction horizon to a single step for linear~\citep{angeli2002} and nonlinear systems \citep{limon2003}. Stability is enforced by an additional constraint that requires an auxiliary optimization to be solved before evaluating the MPC feedback law. {Yet, the underlying optimization problem is nonconvex, for which no runtime guarantees can be given. To tackle nonlinear control problems, small-sized quadratic programs (QPs) based on control barrier and control Lyapunov functions \citep[CBF/CLFs,][]{ames2017} are a popular alternative, i.e., constraint satisfaction and stabilization are provided by a convex optimization.}  A robust version of the CBF/CLF QP is provided by \citep{jankovic2018}.
Note, simultaneous stabilization and constraint satisfaction might be conflicting goals. To alleviate this problem in the CBF-CLF QP approach, a slack variable is typically introduced to trade-off between safety and stability, with more emphasis on the first.  

In this paper, we propose a {computationally} lightweight {yet} robust approximate NMPC {approach} that ensures robust constraint satisfaction and input-to-state stability with respect to external disturbances. {Relying on a single QP}, we fuse the efficient computations of convex optimization with the theoretical {guarantees} of robust MPC schemes.
The results build upon  our recently introduced continuous-time infinitesimal-horizon {nominal} NMPC scheme~\citep{olucak2025a}.  The main contributions of this paper are:
\begin{enumerate}
    \item [a.]{Real-time feasibility:}  An efficient QP formulation for the robust infinitesimal-horizon MPC is {proposed}, guaranteeing robust safety, input-to-state stability, and real-time feasibility.
    \item [b.]{Theoretical guarantees:} Sufficient conditions for {closed-loop robust constraint satisfaction and input-to-state stability are derived}.
    \item [c.]{Practical synthesis:} A nonlinear sum-of-squares optimization method for synthesis of {a safe,  robust, and input-to-state stable approximate MPC is provided}.
\end{enumerate}

We demonstrate our approach in numerical experiments for nonlinear constrained spacecraft control and compare it to other methods for robust control.

\section{Problem Statement}
\label{sec: ProblemStatement}
We consider the continuous-time nonlinear input-affine dynamic system
\begin{flalign}
    \label{eq:system}
  &  \dot x(t) = f(x(t),u(t),w(t)) = f_0(x) + g(x)u + p(x)w&
\end{flalign}
for all $t \geq 0$  with states $x(t) \in \mathbb R^n$ and inputs $u(t) \in \mathcal U$, where $\mathcal U \subset \mathbb R^m$ denotes the viable inputs and $w \in \mathcal{W} \subseteq \mathbb R^{ n_w}$ are additive disturbances, where $f_0: \mathbb R^n \rightarrow \mathbb R^n$, $g: \mathbb R^n \rightarrow \mathbb R^{n \times m}$, $p: \mathbb R^n \rightarrow \mathbb R^{n \times n_w}$.
We assume that $f: \mathbb R^n \times \mathcal{U} \times  \mathcal{W} \to \mathbb R^n$ is Lipschitz continuous, satisfies $f(0,0,0) = 0$, and $\mathcal U$ and  $\mathcal{W} $ are compact sets with $0 \in \interior \mathcal U$ and $0 \in \interior \mathcal W$. For any $t_1 \geq t_0 \geq 0$, we denote by $\mathscr U_{[t_0, t_1]}$ the set of Lebesgue-measurable functions $\uu: [t_0, t_1] \to \mathcal U$. Similar, denote 
$\mathscr W_{[t_0, t_1]}$ the set of Lebesgue-measurable functions $\ww: [t_0, t_1] \to \mathcal W$, i.e., the sequence of bounded yet unknown disturbances.
Let $\xi(t_1, \mathbf u, \ww, t_0, x_0)$ and $\phi(t_1, \kappa,\ww, t_0, x_0)$ be the solutions to the initial value problem of \eqref{eq:system} on $[t_0, t_1]$ with $x(t_0) = x_0 \in \mathbb R^n$, the first for any $\mathbf u \in \mathscr U_{[t_0, t_1]}$, and the second for any state feedback $\kappa:\mathbb{R}^n\rightarrow\Uset$, under disturbance $\mathbf w \in \mathscr W_{[t_0, t_1]}$. If only the nominal system is considered, we simply write $\xi(t_1, \mathbf u, 0, t_0, x_0)$ and $\phi(t_1, \kappa,0, t_0, x_0)$.

{The \emph{tangent cone} of a set $\mathcal A \subset \mathbb R^n$ at $\bar x \in \mathcal A$ is the set $\mathcal T_{\bar x} \mathcal A$ of all vectors $z \in \mathbb R^n$ for which there exists sequences $x_k \in \mathcal A$ and $h_k \in \mathbb R$ with $x_k \to \bar x$, $h_k \searrow 0$, and $(x_k - \bar x)/h_k \to z$,} {where $h_k \searrow 0$ refers to a sequence of strictly positive real numbers that converge to 0.} 

\subsection{Robust Model Predictive Control}
\label{sec:MPC}
Let $\CS, \TS \subset \mathbb R^n$ be the closed state and terminal constraint sets, respectively (with $0 \in \interior \TS$), and $T > 0$ be the prediction horizon. A common choice is to consider the worst case realization of the disturbance, which leads to a min-max MPC formulation, as outlined next. The MPC optimal value function for the state $x_k \in \CS$ is given by {
\begin{subequations}
    \label{eq:mpc-problem}
\begin{flalign}
    &V(x_k) = {} \min_{\mathbf u \in \mathscr U_{[0, T]}} \max_{\mathbf w \in \mathscr W_{[0, T]}} \mathscr J_{[0, T]}(\xi(\cdot, \mathbf u,\ww, 0, x_k), \mathbf u) \\
    &\text{s.t.}\quad \xi(T, \mathbf u, \ww, 0, x_k) \in \TS \\
    &\quad\quad \xi(t, \mathbf u, \ww, 0, x_k) \in \CS \quad \forall t \in [0, T] 
\end{flalign}
 with cost functional
\begin{flalign}
    \label{eq:mpc-cost}
    &\mathscr J_{[0, T]}: (\mathbf x, \mathbf u) \mapsto F(\mathbf x(T)) + \int_0^T L(\mathbf x(t), \mathbf u(t)) \,\mathrm d t&
\end{flalign}
\end{subequations}}
and continuous functions $F: \TS \to \mathbb R_{\geq 0}$ and $L: \CS \times \mathcal U \to \mathbb R_{\geq 0}$ (both positive definite with respect to their respective origins). 
A common choice is 
\begin{flalign*}
    &L(x,u) = x^\top Q x + u^\top R u, 
    \qquad 
    F(x) = x^\top S x&
\end{flalign*}
with positive definite matrices $Q, S \in \mathbb R^{n \times n}$ and $R \in \mathbb R^{m \times m}$.
 {For a given sampling period $\delta < T$, denote $t_{k+1} = t_k + \delta$. For an initial condition $x_0 \in \CS$, we define the MPC feedback law
\begin{flalign}
    \label{eq:mpc-feedback}
    &u_\text{MPC}(t) = \mathbf u_{x_k}^\star(t-t_k), \quad t \in [t_k, t_{k+1}),&
\end{flalign}
where $\mathbf u_{x_k}^\star$ is the optimal solution for $x_k$, and
\begin{flalign*}
    &x_{k+1} = \xi(t_{k+1}, u_\text{MPC}(\cdot), w_k, t_k, x_k)&
\end{flalign*}
is the next sample. In the absence of disturbances in \eqref{eq:mpc-problem}, {Eq.}~\eqref{eq:mpc-feedback} reduces to a nominal NMPC feedback.}

To analyze the nonlinear system \eqref{eq:system} under MPC, we first connect the properties of the system with the feasibility of the optimal control problem  \eqref{eq:mpc-problem}. 
Given a state set $\CS \subseteq \mathbb R^n$, a terminal set $\TS \subseteq \CS$, and an interval $[t_0, t_1] \subset \mathbb R_{\geq 0}$, the robust reach-avoid set is
\begin{multline}
       \label{eq:sys-reachavoid}
        \RA{t_0,t_1}{\CS,\TS} = \{x_0 \in \CS ~|~ \exists \mathbf u \in \mathscr U_{[t_0,t_1]}, \forall \ww  \in \mathscr W_{[t_0, t_1]},  \\ \, \forall t \in [t_0,t_1] \,,
        \xi(t,\mathbf u,\ww,t_0,x_0) \in \CS, \xi(t_1,\mathbf u,\ww,t_0,x_0) \in \TS \}.
\end{multline}
For the MPC problem \eqref{eq:mpc-problem} with prediction horizon $T$ and initial condition $x \in \CS$, the feasible set is
    \begin{multline}
        \label{eq:mpc-feasible}
        \mathscr F(x) = \{ \mathbf u \in \mathscr U_{[0,T]} ~|~\forall \ww \in \mathscr W_{[0, T]}, \forall t \in [0,T] \\\xi(t,\mathbf u,\ww,0,x) \in \CS, \xi(T,\mathbf u,\ww,0,x) \in \TS\} ,
    \end{multline}
    and the set of feasible initial conditions is $\Feas = \{ x \in \CS ~|~ \mathscr F(x) \neq \varnothing \}$.
The following Lemma gives a relationship between the feasible set and the robust reach-avoid-set and was originally proven for the nominal case  in~\citep{Cunis2021}.
\begin{lemma}
    The set of feasible initial conditions $\Feas$ for \eqref{eq:mpc-problem} equals the robust reach-avoid set \eqref{eq:sys-reachavoid} with $t_0 = 0$, $t_1 = T$, i.e.,
        $\Feas = \RA{0,T}{\CS,\TS}$.
\end{lemma}
\begin{proof}
    The proof is analogous to the nominal case as provided by \citep{Cunis2021}.
\end{proof}

{Next, we introduce the notion of input-to-state stability of \citep{sontag2008} and terminal conditions for input-to-state practical stability following \citep{raimondo2009} for min-max MPC.}

\subsection{Input-to-State Stability}
ISS provides a general framework to formulate different notions of stability with respect to inputs (e.g. disturbances). We start with the standard definition~\citep[Section~2.9]{sontag2008}
\begin{definition}
A general nonlinear system
\begin{flalign}
    &\dot x = f(x,w)&
    \label{eq: nonLinSys}
\end{flalign}
is {\em input-to-state stable} for the {(disturbance)} input $w$ if {and only if} there exists a $\mathcal{K}\mathcal{L}$-function $\beta: \mathbb R_{\geq 0}\times \mathbb R_{\geq 0} \rightarrow \mathbb  R_{\geq 0}$ and a $\mathcal{K}$-function $\gamma:  \mathbb R_{\geq 0} \rightarrow \mathbb R_{\geq 0}$ such that for each input signal $\mathbf w$ and each initial condition {$x_0$} it holds that
\begin{flalign}
    &\|\psi(t,\mathbf w,x_0)\| \leq \beta(\|x_0\|,t) + \gamma(\|\mathbf w\|_\infty)&
\end{flalign}
where $\|\cdot\|$ and $\|\cdot\|_\infty$ denote the Euclidean and infinity norm, respectively, and {$\psi(t,\mathbf w,x_0)$ denotes the solution to the initial value problem of \eqref{eq: nonLinSys} on $[0, t]$ with initial condition $x(t_0) = x_0$ under disturbance $\mathbf w$}.
\end{definition}

An ISS system is asymptotically stable if the input is absent or vanishing. Otherwise, if the input is bounded, then the system is ultimately bounded in a set that depends on the input bound. There exist various equivalence {conditions} for a  system to be ISS~\citep{sontag2008}. 

One of them, the notion of {an ISS Lyapunov function, can be extended to systems with control and disturbance inputs:}
A smooth function $V:\mathbb R^n \rightarrow \mathbb R$ is called an {\em ISS control Lyapunov function} (ISS-CLF) for system~\eqref{eq:system} if there exist $\alpha_1, \alpha_2 \in \mathcal{K}_\infty$   such that
\begin{flalign}
    & \alpha_1(\|x\|) \leq V(x) \leq \alpha_2(\|x\|) \label{eq: ISS_CLF1}&
\end{flalign} 
holds {for all $x \in \mathbb R^n$},
and $\alpha, \chi \in \mathcal{K}_\infty$ such that 
\begin{flalign}
  &   \exists u \in \mathcal U, \quad \nabla V(x) f(x,u,w) \leq - \alpha(\|x\|) + \chi(\|w\|)\label{eq: ISS_CLF2}&
\end{flalign} 
holds {for any $x \in \mathbb R^n$ and any $w \in \mathcal{W} \subset \mathbb R^{n_w}$}. Note that, for a given feedback law $u = \kappa(x)$ that meets these conditions, $V$ becomes an ISS-Lyapunov  function for the closed-loop system {given by $f_\kappa: (x, w) \mapsto f(x, \kappa(x), w)$}.

\subsection{Terminal Conditions}
{The terminal constraint set $\mathcal X_T$ together with the terminal penalty $F$ in \eqref{eq:mpc-problem} allow to establish robust stability of the closed-loop dynamics, provided the pair $(F, \mathcal X_T)$ satisfies a suitable terminal condition.}
Since a local control law {$u = \kappa(x)$} is used, we {introduce} $f_\kappa: (x,w) \mapsto f(x,\kappa(x), w)$.

\begin{definition}
    \label{def:terminal-robust}
    {A pair $(P, \Omega)$ with $0 \in \interior \Omega$, $\Omega\subseteq\CS$, and $P: \Omega \to \mathbb R$ satisfies the {\em robust stabilizing terminal condition} for system \eqref{eq:system} if {and only if}} 
    there exists a state feedback $\kappa: \Omega \to \mathcal U$ {along with functions $\{\alpha_1, \alpha_2,\alpha, \sigma \}\in \mathcal{K}_\infty$} such that
\begin{subequations}
        \label{eq:terminalRobust}
    \begin{flalign}
        \alpha_1(\|x\|) &\leq P(x) \leq \alpha_2(\|x\|)         \label{eq:ISS1} \\ 
        \nabla P(x) f_\kappa(x,w) &\leq -L(x, \kappa(x))  - \alpha(\|x\|) + \sigma(\|w\|) \label{eq:ISS2}\\
        f_\kappa(x,w) &\in\Tbdl{x}\Omega\quad 
        \label{eq:terminal-se$T$-robust}
    \end{flalign}
    \end{subequations}
    for all $x \in \Omega$ and $w \in \mathcal{W}$, where $\Tbdl{x}\Omega$ is the tangent {cone} of $\Omega$ at $x\in\Omega$.
\end{definition}

Note that, in the absence of disturbances, these conditions reduce to the nominal terminal conditions~\citep[Assumption 5.9]{grune2017}.
According to \citep[p.~11]{raimondo2009}, the closed-loop system {$\dot x = f_\kappa(x,w)$} is ISS in {the terminal set $\mathcal X_T$ if $(F, \mathcal X_T)$ satisfies the robust terminal condition in Definition~\ref{def:terminal-robust}}. However, for the closed-loop min-max MPC {given by} \eqref{eq:mpc-problem} under state-independent uncertainties, only input-to-state {\em practical} stability can be established \citep[cf.][Corollary~1]{raimondo2009}. 
Subsequently, we propose a scheme that establishes ISS and robust constraint satisfaction for the closed-loop system.

Because robust control barrier functions play a central role in our approach later, we provide a short introduction next.

\subsection{Robust Control Barrier Function}
\label{sec:CLF-CBF}
Invariance plays a central role for safety of dynamical systems and ensures trajectories that emanate from the inside of an invariant (safe) set stay in this safe set for all future times.  Denote $\mathcal A \subset \mathbb R^n$ as $T$-robust control-invariant for \eqref{eq:system} if and only if for all $x_0 \in \mathcal A$, there exists $\mathbf u \in \mathscr U_{[0,T]}$ such that $\xi(t, \mathbf u,\ww, 0, x_0) \in \mathcal A$ for all $\ww \in\mathcal{W}_{[0,T]}$ and $t\in [0,T]$. A set is robust control-invariant if it is $T$-robust control-invariant for all $T>0$. A sufficient condition to ensure \eqref{eq:system} stays in a safe set despite disturbance is given by robust control barrier functions (CBFs).
Let $h:\CS \to \mathbb R$ be a continuously differentiable function that defines a sublevel set $ \mathcal{H}=\{x \in \CS ~|~ h(x) \leq 0\}$. Then $h(\cdot)$ is a robust CBF for \eqref{eq:system} if there exists an extended class-$\mathcal K$ function {$\gamma$}, i.e., a strictly increasing function $\gamma:\mathbb{R}_{\geq 0}\to\mathbb{R}_{\geq 0}$ that satisfies $\gamma(0)=0$,  such that 
    \begin{flalign}
        \label{eq:CBFcondition}
       & \exists u \in \mathcal U, \quad \nabla h(x) f(x,u,w) \leq \gamma(-h(x)) &
    \end{flalign}
for all {$x \in \mathcal H$ and} $w \in \mathcal{W}$.
The existence of a robust CBF ensures that the sublevel set defined by $\mathcal H$ is robustly forward-invariant. Note that a (robust) CBF is frequently defined in terms of a superlevel set.

Both CLFs and CBFs are widely used to ensure stability and safety. In practice, they are often employed together. In such a case, the functions must be \emph{compatible}, i.e., for every $x \in \CS$ and $w \in \mathcal W$, there must exist a $u \in \mathcal U$ such that \eqref{eq: ISS_CLF1}, \eqref{eq: ISS_CLF2} and \eqref{eq:CBFcondition}  hold simultaneously.

\section{Methodology}
\label{sec:Methodology}
{In this section, we derive a real-time feasible NMPC scheme which ensures robust constraint satisfaction and {input-to-state} stability. We thereby build upon the recently introduced infinitesimal-horizon MPC scheme ($\partial$MPC), for which we provide a summary of the core results of \citep{olucak2025a} first. Afterwards, we provide the main theoretical results of this paper.}

\subsection{Infinitesimal-horizon MPC Scheme}
Recall, in the absence of disturbances, the conditions in Definition~\ref{def:terminal-robust} reduce to the nominal terminal conditions~\citep[Assumption 5.9]{grune2017}. If these conditions are met, the feasible set (denoted by $\mathcal{R}$) of the nominal MPC with horizon $T > 0$ is rendered forward invariant and the optimal value function (denoted by $V$) is non-increasing. However, for the MPC to be feasible, a potentially long horizon is  needed \citep[cf.][Remark 4.1]{chen1998}, which in turn results in a high computational effort. 

{The infinitesimal-horizon MPC scheme replaces the optimization over a (potentially very long) horizon by suitable approximations $\hat V: \mathbb R^n \to \mathbb R_{\geq 0}$ and $\hat {\mathcal R} \subset \mathcal X$ of the optimal value function and feasible set, respectively.}
Assuming a quadratic stage cost $L$, the $\partial$MPC problem for $\uuinfhat: \Rap\to\Uset$ is given by the QP~\citep{olucak2025a}
\begin{subequations}
\label{eq: infQP}
\begin{flalign}
    \hat\kappa_\partial: x \mapsto {\hat u} \in {} \operatorname*{arg\,min}_{u \in \mathcal{U}} \; L(x, u) +  {\nabla \Vap(x) f(x,u)}  \\
    \text{s.t. ${\nabla \hat h(x) f(x,u)} \leq \gamma( -\hat h{(x)} )$}
\end{flalign}
\end{subequations}
{where $\hat h: \mathbb R^n \to \mathbb R$ satisfies $\Rap = \{x\in\CS~|~\hat h(x) \leq 0\}$,}
for any given, sampled state $x \in \Rap$ and where $\gamma(\cdot)$ is an extended class-$\mathcal{K}$ function.

\Citet[Theorem~1]{olucak2025a} demonstrated that, if $(\hat V, \hat h)$ are a compatible CLF/CBF pair that satisfies the nominal terminal conditions, the feedback law \eqref{eq: infQP} preserves feasibility and asymptotic stability.
{Moreover, it was found that~\citep[Section~III]{olucak2025a} any compatible CBF/CLF pair $(\hat h, \hat V)$ safely stabilizes the nominal system, i.e., a compatible CBF/CLF pair is an approximation of the optimal value function and feasible set, respectively.}

\subsection{Robust infinitesimal-horizon MPC}
In the {spirit of the} $\partial$MPC {scheme}, this subsection provides sufficient conditions to ensure robust constraint satisfaction and {input-to-state stabilization by an infinitesimal-horizon MPC feedback}. 
Assume a quadratic stage cost $L$ and  let $(\hat V, \hat h)$ be {approximations of the optimal value function and feasible set of the min-max optimization problem \eqref{eq:mpc-problem}, respectively}. Subsequently, we assume disturbances of the form
\begin{flalign}
    \label{eq:disturbance-bound}
    \mathcal{W} := \{w \in \mathbb R^{n_w} \mid \lvert\lvert w \rvert\rvert_p \leq \bar w \}
\end{flalign} 
for some $p \in \{1, 2, \infty\}$, where $\bar w > 0$ denotes the peak disturbance.
The ISS-$\partial$MPC {aims to solve the infinitesimal-horizon min-max optimization problem}
\begin{subequations}
\begin{align}
    \min_{u \in \mathcal U} \max_{w \in \mathcal W} L(x,u) + \nabla \hat V(x) f(x,u,w) \\
    \text{s.t. $\nabla \hat h(x) f(x,u,w) \leq \gamma(-\hat h(x))$}
\end{align}
\end{subequations}
for any given, sampled state $x \in \Rap :=\{x\in\CS~|~\hat h(x) \leq 0\}$,  where $\gamma(\cdot)$ is an extended class-$\mathcal{K}$ function.
{Given the structure of $f$ in \eqref{eq:system}, the cost function is jointly linear in $u$ and $w$. Moreover, using the disturbance bound in \eqref{eq:disturbance-bound}, the effect of $w$ on the dissipation inequality can be upperbounded using the operator norm $\|\cdot\|_\text{op}$ induced by $\|\cdot\|_p$ to $\delta(x) = \|{\nabla \hat h(x) p(x)}\|_\text{op} \bar{w}$.}

{The infinitesimal-horizon problem hence reduces to} 
\begin{subequations}
    \label{eq: infQPRobust}
\begin{flalign}
    \hat\kappa_\partial: x \mapsto {\hat u} \in {} \operatorname*{arg\,min}_{u \in \mathcal{U}} \; L(x, u) +  {\nabla \Vap(x) [f_0(x) + g(x) u]}  \\
    \text{s.t. ${\nabla \hat h(x) [f_0(x) + g(x) u]} + \delta(x) \leq \gamma( -\hat h{(x)} )$}  \label{eq:mpc-infinitesimal-constraint_robust}
\end{flalign}
\end{subequations}
for any $x \in \hat {\mathcal R}$, {which is a quadratic program if the stage cost $L$ is quadratic in $u$.}

\begin{theorem}
    \label{thm:inf-mpc-robust}
    Suppose $(\hat V,\hat h)$ satisfies the stabilizing terminal conditions from Definition~\ref{def:terminal-robust}, that is, $\hat V$ is an ISS-CLF and $\hat h$ is a robust CBF. Let $\phi(t,\hat \kappa_\partial,\ww,t_0,x_0)$ denote the closed-loop solution under $\hat \kappa_\partial$ for any $\ww \in \mathscr W_{[t_0, t_1]}$.
    Then {there exist $\beta, \mu \in \mathcal{K}_\infty$ such that the following assertions hold} for any $x_0 \in \Rap$ and $\ww \in \mathscr W_{[t_0, t_1]}$:
    \begin{enumerate}
        \item {The robust $\partial$MPC problem \eqref{eq: infQPRobust} is feasible, that is,} 
        \begin{align*}
            \phi(t,\kappa_\partial,\ww,0,x_0) \in \Rap
        \end{align*}
        for all $t \geq 0$. 
        \item {The approximation $\hat V$ is an ISS-Lyapunov function for the closed loop under $\hat \kappa_\partial$, that is,}
        \begin{align*}
            \hat V(\phi(t, \hat \kappa_\partial, \ww, 0, x_0)) \leq \beta(\|x_0\|, t) + \mu(\|\ww\|_\infty)
        \end{align*}
        for all $t \geq 0$.
    \end{enumerate}
    {Consequently, the feedback $\hat \kappa_\partial$ is defined for all times $t \geq 0$, the closed-loop trajectories $\phi(t, \hat \kappa_\partial, \ww, 0, x_0)$ robustly satisfy the constraints $\mathcal X$, and the closed-loop system is input-to-state stable to disturbance $w \in \mathcal W$ on $\Rap$.}
\end{theorem}
\begin{proof}
The first {assertion} (feasibility) is an immediate consequence of $\hat \kappa_\partial(x) = \hat u$ satisfying \eqref{eq:mpc-infinitesimal-constraint_robust}, since ${\nabla \hat h(x) p(x)w} \leq \|{\nabla \hat h(x) p(x)}\|_2 \bar{w} = \delta(x)$ ensures robust constraint satisfaction for the worst-case realization of the disturbance. 
For the second {assertion} (ISS), recall that $\hat V(x)$ {satisfies \eqref{eq:ISS1} and \eqref{eq:ISS2}} {for some} $\kappa(x)$ for all $w \in \mathcal{W}$ and for all $x \in \Rap$ according to Definition~\ref{def:terminal-robust}. Thus,
we have 
\begin{multline*}
    L(x, \hat\kappa_\partial(x)) + {\nabla \Vap(x) f(x,\hat\kappa_\partial(x),{w})} \\ 
    \leq L(x, \kappa(x)) +  {\nabla \Vap(x) f(x,\kappa(x),{w})} \\
    \leq  - \alpha(\|x\|) + \sigma(\|w\|)
\end{multline*}
{where the first inequality follows from the fact that $\hat \kappa_\partial(x) = \hat u$ minimizes \eqref{eq: infQPRobust} independently of $w$. Reordering the inequality shows that $\hat V$ is an ISS-Lyapunov function for the closed loop of \eqref{eq:system} and $\hat\kappa(x)$. Noting that $\hat {\mathcal R} \subset \mathcal X$}
concludes the proof. \qed
\end{proof}

{While} the approach of \citep{raimondo2009} only ensures input-to-state practical stability of the closed loop under min-max MPC feedback, our approach provides input-to-state stability and recovers asymptotic stability in case that disturbances are absent or vanish. This is  because the dissipation rate of the ISS-CLF is minimized only for possible, current disturbances instead of future disturbance signals. {Note {that}, unlike \textit{standard} robust CBF-CLF formulations \citep[e.g.,][]{jankovic2018}, our formulation can incorporate stage cost, no slack variable is needed to ensure feasibility of the underlying QP and input-to-state stabilization is realized via a terminal penalty instead of a constraint.}

Finding a compatible CBF/CLF pair that meets  Definition~\ref{def:terminal-robust} might be challenging. It might be accomplished with Hamilton-Jacobi reachability or sum-of-squares (SOS) programming. In this paper, we adapt the SOS method of \citet{olucak2025a} in the next section.

\section{Polynomial Synthesis}
\label{sec: SOSApprox}
In this section, we leverage SOS programming to compute polynomial approximations $\hat V$ and $\hat {\mathcal R}$. Hence, a short introduction to SOS is given first, followed by {a nonconvex SOS optimization problem for the} synthesis.

\subsection{Sum-of-Squares Programming}
Denote the set of polynomials in $x$ with real coefficients up to degree $d$ by $\mathbb R_d[x]$. 
A polynomial $p \in \mathbb R_{2d}[x]$ is a {\em sum-of-squares polynomial} ($p \in \Sigma_{2d}[x]$) if and only if there exist $m \in \mathbb N$ and $p_1, \ldots, p_m \in \mathbb R_d[x]$ such that $p = \sum_{i=1}^m (p_i)^2$. If $p \in  \Sigma_{2d}[x]$, then $p(x) \geq 0 $ for all $x \in \mathbb{R}^n$. Convex SOS program{s} can be {reduced to} and solved as semidefinite programs (SDPs) \citep{parrilo2003}. For synthesis, one encounters nonconvex sum-of-squares problems of the form
\begin{flalign}
    \label{eq:sos-nonlinear}
    \min_{\xi \in \mathbb R_{2d}[x]^n} f(\xi) \quad \text{s.t. $\xi \in \Sigma_{2d}[x]^n$ and $g(\xi) \in \Sigma_{2d'}[x]^m$}
\end{flalign}
where $f: \mathbb R_{2d}[x]^n \to \mathbb R$ and $g: \mathbb R_{2d}[x]^n \to \mathbb R_{2d'}[x]^m$ are differentiable functionals. Nonconvex SOS problems need iterative schemes such as coordinate-descent~\citep{chakraborty2011}, bisections~\citep{seiler2010a}, or sequential SOS~\citep{cunis2023}.

\subsection{Synthesis Problem}
In this subsection, we provide a nonconvex SOS problem to synthesize a compatible pair of robust CBF and ISS-CLF. A full derivation using set-inclusion constraints and the generalized $\mathcal{S}$-procedure~\citep{wang2023} is given in the supplementary material. 
We assume that the system dynamics $f$ in \eqref{eq:system} and the stage cost $L$ in \eqref{eq:mpc-problem} are polynomial functions, that the state constraint set (or a tight inner approximation) is a basic semialgebraic set
\begin{flalign*}
    \CS = \{x \in \mathbb R^n ~|~ \ell(x) \leq 0 \}
\end{flalign*}
for some $\ell \in \mathbb R[x]$, and that the input constraint set is the polyhedral set
\begin{flalign*}
    \Uset = \{ u \in \mathbb R^m ~|~ H_\Uset u \leq \mathbf 1_p \}
\end{flalign*}
where $H_\Uset \in \mathbb R^{p \times m}$, $p \in \mathbb N$, and $\mathbf 1_p$ denotes the vector $(1,\ldots,1) \in \mathbb R^p$.  
Denote the feasible set approximate as
\begin{align*}
    \mathcal{\hat{R}}:= \{x \in \mathbb R^n \mid \hat h(x) \leq \beta \} \subseteq \mathcal X
\end{align*}
for some level $\beta > 0$ and {problem} $\hat h \in \mathbb R[x]$.

We are interested in synthesizing a robust CBF $\hat h \in \mathbb R_{2d}[x]$,  and an ISS-CLF $\hat V \in \mathbb R_{2d}[x]$ that are simultaneously compatible. 
The state dependence is omitted if clear from context. Since a control law $\kappa(x)$ is synthesized, we can write $f_\kappa(x,w) = f(x,\kappa(x), w)$ {as well as}  $\dot V_\kappa(x, w) = \nabla \hat V(x) f_\kappa(x, w)$ and {${\dot{\hat{h}}_\kappa(x, w) = {\nabla \hat h(x) f_\kappa(x,w)}}$}.  For notational convenience, we introduce
\begin{align*}
 {\tau(V,\kappa) := \dot V(x,w)  + L(x,\kappa(x)) - \sigma(\lvert\lvert w  \rvert\rvert)   + \alpha(\lvert\lvert x  \rvert\rvert ) \leq 0 }
\end{align*}
and note that to synthesize the class-$\mathcal{K}_\infty$ functions, we make use of the method proposed by~\citep[Section 4.1]{ichihara2012a} and {realize} these as univariate real even polynomials in ${r} \in \mathbb R$.  {{The last step of} the SOS synthesis {is} the $p$-norm bound \eqref{eq:disturbance-bound} of the disturbances. We restrict the {following} formulation to the  2-norm, which can be encoded {using $\|w\|_2^2 = {\langle w, w \rangle}$}. 
\begin{remark}
{Other} $p$-norms can be incorporated into the SOS program \eqref{eq: SOS_Syn_robMPC} {as well}. The 1-norm {and the} $\infty$-norm can be {included as} intersections of element-wise inequalities or {replaced} by polynomial over-approximations. Note {that} this might also effect the dimensions of the SOS multiplier in \eqref{eq: SOS_Syn_robMPC}.
\end{remark}}

The final  SOS program for the synthesis reads
\begin{subequations}
    \label{eq: SOS_Syn_robMPC}
\begin{flalign}
    &\min_{\substack{s_1, s_4 \in \Sigma[x], s_3 \in\Sigma[x]^p, \\  s_2, s_{2,w}, s_5, s_{5,w} \in \Sigma[x,w], \\
    \hat h , \hat V \in \mathbb{R}[x], \kappa \in \mathbb{R}[x]^m, \\\alpha_1,\alpha_2, \alpha, \sigma \in \mathbb R[s]}}    J(\hat V, \hat h, \kappa)   \\
    &\text{subject to} \nonumber \\
    & \hat V - \alpha_1(\lvert\lvert x \rvert\rvert)                                  \in\Sigma[x]                          \label{eq: strictPos1}\\
    &   \alpha_2(\lvert\lvert x \rvert\rvert) -\hat V                                                         \in\Sigma[x]  \label{eq: strictPos2}\\
    &   s_1 (\hat h-\beta) - \ell                                                                              \in\Sigma[x] \\
    & (s_2-a) (\hat h-\beta) -  \dot{\hat{h}} + s_{2,w}  (\|w\|{{_2^2}} - \bar{w}{^2})                         \in\Sigma[x,w] \label{eq: CBFSOS}\\
    &  s_3 (\hat h-\beta) - (H_\mathcal{U} \hat\kappa - \mathbf 1_p)                                          \in\Sigma[x]^p \label{eq: lowerContSOS} \\
    &  s_4 (\hat h-\beta) - {\tau(V,\kappa)}   + s_{4,w}(\|w\|{{_2^2}} - \bar{w}{^2})                                           \in\Sigma[x,w] \label{eq: terminalPenSOS}\\
    &  r \frac{\mathrm{d} \zeta(r)}{\mathrm{d} r}              \in\Sigma[r]  \quad\text{ for }\zeta \in \{\alpha_1,\alpha_2, \alpha, \sigma\}  \label{eq: Kinf1}
\end{flalign}
\end{subequations}
 where in~\eqref{eq: CBFSOS} the class-$\mathcal{K}$ function is of the form $\gamma(r) = a r$ with $a > 0$. Condition~\eqref{eq: Kinf1} ensures that $\zeta(s) $ is strictly increasing for $ s \geq 0$, $\zeta(s_1) < \zeta(s_2)$ for $0 \leq s_1  < s_2$. 
{The solution $(\hat V, \hat h)$ of the nonlinear SOS program \eqref{eq: SOS_Syn_robMPC} satisfies the assumptions in Theorem \ref{thm:inf-mpc-robust} and hence, the closed-loop system is input-to-state stable with respect to external disturbances {and ensures robust constraint satisfaction}.}

\section{Numerical Results}
This section provides a numerical comparison {of the robust $\partial$MPC scheme (ISS-$\partial$MPC)} to zoRo-RTI~\citep{frey2024}, the polynomial control law $\tilde \kappa$ from~\eqref{eq: SOS_Syn_robMPC}, and a standard controller~\citep[Subsection~5.2.1]{fichter2023} for constrained spacecraft control {along with} a detailed analysis of the proposed approach.
For the computation as described in Section~\ref{sec: SOSApprox}, we make use of {\rm Ca$\Sigma$oS}~\citep{cunis2025a} v1.0.0 with {\rm MOSEK}~\citep{mosek} v11.1.5 as the underlying SDP solver. 
For zoRo-RTI, we make use of \texttt{acados}~\citep{verschueren2022} v0.5.3 and \texttt{osqp}~\citep{stellato2020} with partial condensing. Whereas \texttt{acados} provides a generic interface for all kinds of optimal control problems, a direct native QP interface is not available. Therefore, we make use of {\rm CasADi}~\citep{andersson2018} v3.6.7 to setup the ISS-$\partial$MPC QP with  \texttt{osqp}.
The computation time (wall-time) is provided by the CasADi solver statistics for the ISS-$\partial$MPC {scheme} and we measure the wall-time for zoRo-RTI via MATLAB's \texttt{tic-toc} command. Source code for the synthesis and the simulations along with details about the zoRO-RTI implementation\footnote{Our implementation builds upon \texttt{zoro\_example.m} provided in \texttt{acados}.}  are
provided in the supplementary material~\citep{olucak2026b}. 
The results in this paper were computed on a  computer with Windows 11 and MATLAB 2026a running on an AMD Ryzen 9 5950X 16-Core Processor with 3.40 GHz and 128 GB RAM. 

\subsection{Comparative Study}
We consider the problem of rate damping for a spacecraft from large angular rates. Such a situation can occur after separation from a launcher. 
The state vector {$x \in \mathbb R^3$ corresponds to} the angular rates.
The system dynamics {in continuous time are given by a vector field} $f: \mathbb R^3 \times \mathbb R^3 \times \mathbb R^3 \rightarrow \mathbb R^3$ {with}
\begin{subequations}
\label{eq:spacecraft-dynamics}
\begin{flalign}
\dot{x} = f(x,u,w) =   {-J^{-1} C(x) J x} + J^{-1} (u+w)  
\end{flalign}
where $J \in \mathbb R^{3 \times 3}$ is the inertia tensor, $u \in \mathbb R^3$ and $w \in \mathbb R^3$  are the control and disturbance torques respectively, and {$C(x)$} is the cross-product matrix of the angular rates. The inertia tensor  is assumed to be $J = \text{diag}(31, 77, 78)$ \unit{\kilogram\meter{^2}}. We consider the state and control constraints and disturbance torques \begin{flalign}
    &{x} \in [-0.5, 0.5] \times  [-0.2, 0.2] \times [-0.2, 0.2]~\si{ \unit{\deg\per\second}} \label{eq:rateCon}\\
    &\| u \|_\infty \leq \SI{0.012}{\newton\meter} \\
    & \|w\|_2 \leq \SI{1.2e-3}{\newton\meter} \label{eq: compNoiseMag}
\end{flalign}
\end{subequations}
where \eqref{eq: compNoiseMag} corresponds to ten percent of the available actuation capabilities. {Note that in \eqref{eq: compNoiseMag} we assume disturbances of the form
$\mathcal{W} := \{w \in \mathbb R^{n_w} \mid \lvert\lvert w \rvert\rvert_2 \leq \bar w \}$.}

\subsubsection{Preparation}
The simulation time is set to \SI{250}{\second} and a sampling frequency of \SI{10}{Hz}.  We draw 20 randomly sampled initial angular rates from a shifted and scaled uniform distribution such that they lie in the precomputed feasible set of the ISS-$\partial$MPC, that is, $\hat h(x_\mathrm{sample}) \leq 0$.
We assume $Q = \mathrm{diag}(10,10,10)$ and $R = \mathrm{diag}(0.1,0.1,0.1)$ for both, the ISS-$\partial$MPC and zoRo-RTI. For zoRo-RTI, we determine a terminal penalty for the linearized system $F: x \mapsto x^\top S x$, where $S \in \mathbb R^{3 \times 3}$ is the solution of the Riccati for $Q$ and $R$. We find the maximum stable level set $\gamma  =  4.4\cdot 10^{-5}$ for the nonlinear constrained system using a SOS program via bisection and define $\TS = \{ x \in \mathbb R^3 \, | \, F(x) \leq \gamma \}$. A prediction horizon of $T = \SI{20}{\second}$ was heuristically determined such that the first step is feasible and discretization is chosen such that $\mathrm{d}t = \SI{0.1}{\second}$.
The standard rate controller reads {$K(x) = - K_c x$}, with {$K_c = c J$} where $c >0 $ is a tuning parameter~\citep[Subsection~5.2.1]{fichter2023}. This control technique is not a robust scheme, but a common control technique for the nonlinear dynamics with some inherent robustness. Furthermore, we add the polynomial control law from the synthesis step to the comparison. 
{The disturbance realizations are drawn uniformly such that each sample $w_k \in \mathcal{W}$.}

\subsubsection{Results}
We compare the four approaches in terms of mean {and maximum} computation time, precomputation time, root mean square (RMS)  value  of the angular rate error and integral stage cost of the closed-loop trajectories. The results are presented in the aggregated Table~\ref{tab:comparsionRateControl} over the 20 samples, marking the best in each category \textbf{bold}. The precomputation times for all methods is just a few seconds. For the two analytical control laws we do not provide computation times. The proposed {ISS-$\partial$MPC} approach is significantly faster compared to zoro-RTI, while both have the almost same average stage cost, with zoro-RTI slightly better. A ratio of the mean RMS value between the three other approaches and the ISS-$\partial$MPC (RMS/$\text{RMS}_{\mathrm{ISS}-\partial \mathrm{MPC}}$) is provided for $ t \in [\SI{100}{\second},\SI{250}{\second}]$, indicating that the other methods have larger errors on average. For the first run, we provide the RMS over time in Fig.~\ref{fig:ComparisonRObustMPC} for a better visual impression.

\begin{figure}
\setlength{\figH}{3.8cm}
\setlength{\figW}{7cm}
    \centering
    \includegraphics{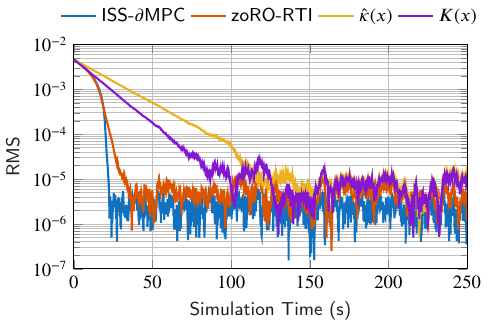}
    \caption{RMS of the angular rate for the first run. Among the four methods, the proposed ISS-MPC has the smallest RMS.}
    \label{fig:ComparisonRObustMPC}
\end{figure}

\begin{table}
    \centering
    \caption{Aggregated comparison of the different controllers in terms of precomputation time $t_\text{pre}$, mean (worst-case) computation time $t_\text{comp}$, RMS ratio  and integral stage cost.}
    \label{tab:comparsionRateControl}
    \begin{tabular}{lccrlc}
        \hline
    \hline
    Method & $t_\text{pre}$ [s]&  {$t_\text{comp}$ [ms]} & RMS ratio & $\int L(\cdot) \mathrm dt$ \\
    \hline
     ISS-$\partial$MPC           &  2.65 & $\mathbf{0.015}$ (0.064)  &  $\mathbf{1}$    & $0.0118$\\
     zoRo-RTI                &   $\mathbf{1.9}$    &  6.26 (289.98) &   2.03& $\mathbf{0.0111}$\\
     $\hat\kappa(\cdot)$     & 2.65  &  {---}   & 5.64 &  0.0211\\
     $K(\omega)$             &  {---} &  {---} & 3.73 & 0.015\\
         \hline
         \hline
    \end{tabular}
\end{table}

\subsection{Constrained three-axis attitude control}
In the second case, we consider constrained large angle three axis attitude control for a spacecraft to demonstrate that the sufficient condition for ISS and robust constraint satisfaction are met.  We augment the state to {$x = (x_1,x_2)$}, where {$x_1 \in \mathbb R^3$ corresponds again to the angular rates and $x_2 \in \mathbb R^3$ corresponds to} the modified Rodrigues parameters. The system dynamics read
\begin{flalign}
\dot{x} = f(x,u,w) = \begin{bmatrix}  {-J^{-1} C(x_1) J x_2} + J^{-1}(u + w)  \\  {\frac{1}{4} B(x_2) x_1} \end{bmatrix} 
\end{flalign}
with {$B(x_2) = (1-x_2^\top x_2^{}) I_{3\times 3} + 2 C(x_2) + 2 x_2^{} x_2^\top$},
where {$C(x_2)$ is again} the cross-product matrix. An attitude constraint {$x_2^\top x_2^{} \leq 1$} is considered in addition to \eqref{eq:rateCon}-\eqref{eq: compNoiseMag}. 

\subsubsection{Preparation}
We consider rest-to-rest profiles (angular rates zero) and draw 20 random initial attitudes from a shifted and scaled uniform distribution such that they lie in the precomputed feasible set of the ISS-$\partial$MPC, that is, $\hat h(x_\mathrm{sample}) \leq 0$ and set the simulation for $\SI{3600}{\second}$ to investigate the behavior for larger $t$. We choose $Q = \mathrm{diag}(10,10,10,1,1,1)$ and $R = \mathrm{diag}(200,200,200)$. The precomputation time with a reasonable initial guess for the nonconvex SOS problem \eqref{eq: SOS_Syn_robMPC} is $\SI{10}{\second}$. {The disturbance realizations are drawn uniformly such that each sample $w_k \in \mathcal{W}$.}

\subsubsection{Results}
In Figure~\ref{fig: satAttEuler}, we provide the Euler angles {computed from $x_2$} over time for each of the 20 runs. The results show that the state converges to the origin even for large initial Euler angles, indicating a large region-of-attraction. In Figure~\ref{fig: satAttControl}, we provide an excerpt of the control inputs for the first 30 seconds of the closed-loop trajectories. The control inputs start fairly aggressive but stay within the bounds for the whole time. As indicated in Figure~\ref{fig: satAttCBF}, the CBF evaluated along the closed-loop trajectory is non-positive, indicating that the state constraints are satisfied. The zoom in the right  of Figure~\ref{fig: satAttCBF} shows the CBF values get close to the boundary of the feasible set, but does not violate it. {In Figure~\ref{fig: satAttLyap}, the {ISS-}Lyapunov function evaluated along the closed-loop trajectories is depicted, which {converges (close) to zero}.} The mean and worst computation time over all twenty runs are $\SI{0.015}{\ms}$ and $\SI{0.53}{\ms}$, respectively. 

\begin{figure}
\setlength{\figH}{4cm}
\setlength{\figW}{7cm}
    \centering
     \includegraphics{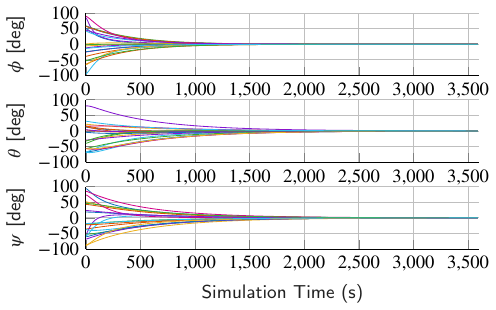}
    \caption{Euler angles for the twenty runs for varying initial conditions for large angle slews.}
    \label{fig: satAttEuler}
\end{figure}

\begin{figure}
\setlength{\figH}{4.5cm}
\setlength{\figW}{7cm}
    \centering
    \includegraphics{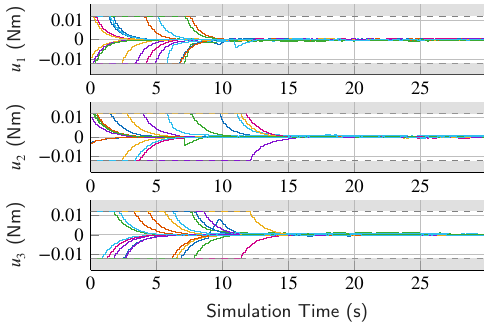}
    \caption{Excerpt of control inputs for the first 30 seconds of the closed-loop trajectory for the twenty runs. The control input stays within the bounds, indicating that the input constraints are satisfied.}
    \label{fig: satAttControl}
\end{figure}

\begin{figure}
\setlength{\figH}{3.8cm}
\setlength{\figW}{7cm}
    \centering
    \includegraphics{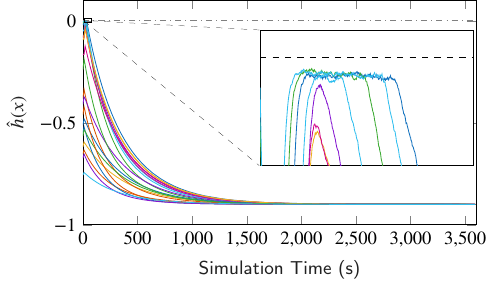}
    \caption{CBF evaluated along the closed-loop trajectory for the twenty runs. The CBF is non-positive, indicating that the state constraints are satisfied. The zoom in the right figure shows the CBF values get close to the boundary of the feasible set, but do not violate it.}
    \label{fig: satAttCBF}
\end{figure}


\begin{figure}
\setlength{\figH}{3.8cm}
\setlength{\figW}{7cm}
    \centering
    \includegraphics{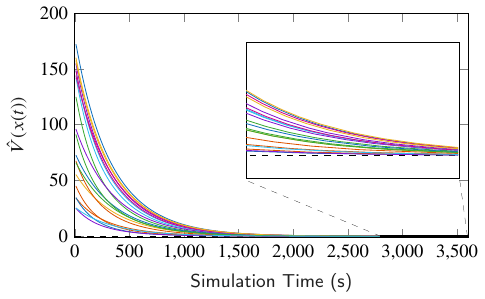}
    \caption{{ISS-}Lyapunov function evaluated along the closed-loop trajectory for the twenty runs. The ISS-CLF is nonnegative, indicating that the state converges to {a neighbourhood of} the origin.}
    \label{fig: satAttLyap}
\end{figure}

\subsection{Discussion}
The proposed ISS-$\partial$MPC guarantees robust constraint satisfaction and input-to-state stability by a precomputed polynomial approximations for the robust CBF and ISS-CLF. Its demonstrated low computational footprint makes it a suitable candidate for controlling applications with severely limited computing resources. 
The solution quality (size of feasible set, control performance) highly depends on the polynomial degrees of
the approximations. Compared to other robust NMPC scheme, the
offline tuning effort is higher because the polynomial structure
(degree, monomials) must be selected in addition to the state
cost and its weights. Selecting the extended class-$\mathcal{K}$ functions in the CBF
affects the performance is an ongoing research activity \citep[see, e.g.,][]{parwana2025}. Similar, {using the} method of \citep{ichihara2012a} {for synthesis of class-$\mathcal K_\infty$ functions} might lead to some conservatism. The low computational effort can be primarily traced back to the offline precomputation step. However, the precomputed robust terminal ingredients are only valid for one equilibrium point. Note that this is also the case for (robust) NMPC with
stabilizing terminal conditions. To alleviate this problem, one could consider reference-dependent terminal conditions \citep{cotorruelo2021a}.
The proposed approach is executed in a quasi-continuous manner, i.e., a continuous time optimization problem is solved in a sampled data system and in between sampling instants the controller is kept fix. This might violate the invariance property and is known as inter-sampling effect~\citep{breeden2022}. {However,} this effect was not {observed} in the numerical examples. {Another interesting direction lies in the development of disturbance observers in combination with ISS-$\partial$NMPC formulations parameterized in the disturbance, to reduce the conservatism of the assumption of worst-case disturbances \citep[cf.][]{jankovic2018}. {The infinitesimal-horizon approach makes this strategy} possible, because current disturbances instead of future disturbance signals are used.}

\section{Conclusions}
An ISS-MPC scheme with a small computational footprint is proposed based on a compatible pair of robust CBF and ISS-CLF. This computational lightweight approximate NMPC scheme comes with theoretical guarantees regarding robust constraint satisfaction and input-to-state stability. Its effectiveness and also the theoretical guarantees are demonstrated in numerical examples. Unlike other approaches that only guarantee {practical stability in case of vanishing or decaying disturbances}, the proposed {ISS-$\partial$MPC scheme} recovers asymptotic stability.

Future work includes investigations to tune of the extended class-$\mathcal{K}$ functions, investigating the use of  (robust) reference-dependent terminal conditions  and {mitigations of} the inter-sampling effect.

\printcredits
\bibliographystyle{cas-model2-names}
\bibliography{literature}

\end{document}